\documentclass[a4paper,10pt]{amsart}
\usepackage[utf8]{inputenc}
\usepackage{amsthm}
\usepackage{amsmath}
\usepackage{bm}
\usepackage{amsfonts}
\usepackage{amssymb}
\usepackage{mathtools}
\usepackage{mathrsfs}
\usepackage{setspace}
\usepackage{xcolor}
\usepackage{textgreek}
\usepackage{thmtools} 

\usepackage[obeyFinal,textwidth=2.7cm]{todonotes}

\usepackage[pdfdisplaydoctitle,colorlinks,breaklinks,urlcolor=blue,linkcolor=blue,citecolor=blue]{hyperref} 

\newcommand{\C}{\mathbb{C}}
\newcommand{\D}{\mathcal{D}}

\newcommand{\Es}{\mathscr{E}}
\newcommand{\F}{\mathcal{F}}

\renewcommand{\H}{\mathcal{H}}

\renewcommand{\L}{\mathcal{L}}

\newcommand{\N}{\mathbb{N}}

\newcommand{\R}{\mathbb{R}}
\renewcommand{\S}{\mathcal{S}}
\newcommand{\T}{\mathbb{T}}
\newcommand{\V}{\mathcal{V}}

\newcommand{\Z}{\mathbb{Z}}

\DeclareMathOperator{\trace}{Tr}

\renewcommand{\epsilon}{\varepsilon}

\newcommand{\ubar}[1]{\underline{#1}}

\DeclareMathOperator{\imm}{i}

\newcommand{\de}{\partial}

\newcommand{\set}[1]{\left\{#1\right\}}
\newcommand{\pa}[1]{\left(#1\right)}

\newcommand{\abs}[1]{\left|#1\right|}
\newcommand{\norm}[1]{\left\|#1\right\|}
\newcommand{\brak}[1]{\left\langle#1\right\rangle}

\newcommand{\expt}[1]{\mathbb{E}\left[#1\right]}
\newtheorem{thm}{Theorem}[section]

\newtheorem{lem}[thm]{Lemma}

\newtheorem{prop}[thm]{Proposition}

\theoremstyle{remark}
\newtheorem{rmk}[thm]{Remark}

\numberwithin{equation}{section}

\usepackage[obeyFinal,textwidth=2.7cm]{todonotes}
\newcounter{mrnotanum}
\title[Mean Field Decay of Correlation for Point Vortices]{Decay of Correlation Rate in the Mean Field Limit of Point Vortices Ensembles}
\author[F. Grotto]{Francesco Grotto}
  \address{Scuola Normale Superiore, Piazza dei Cavalieri, 7, 56126 Pisa, Italia}
  \email{\href{mailto:francesco.grotto@sns.it}{francesco.grotto@sns.it}}
\author[M. Romito]{Marco Romito}
  \address{Dipartimento di Matematica, Universit\`a di Pisa, Largo Bruno Pontecorvo 5, 56127 Pisa, Italia}
  \email{\href{mailto:marco.romito@unipi.it}{marco.romito@unipi.it}}
  \urladdr{\url{http://people.dm.unipi.it/romito}}
  \thanks{The second author acknowledges partial support of University of Pisa through project PRA 2018\_49.}
  \dedicatory{It is a very special pleasure and honour for us to contribute to the volume collecting the proceedings of the conference celebrating the 60\textsuperscript{th} birthday of Franco Flandoli.}
\date\today

\begin{document}

\begin{abstract}
 We consider the Mean Field limit of Gibbsian ensembles of 2-dimensional point vortices
 on the torus. It is a classical result that in such limit correlations functions
 converge to 1, that is, point vortices decorrelate: we compute the rate
 at which this convergence takes place by means of Gaussian integration techniques,
 inspired by the correspondence between the 2-dimensional Coulomb gas and the 
 Sine-Gordon Euclidean field theory. 
\end{abstract}

\maketitle

\section{Introduction}

Mean Field scaling limits of 2-dimensional Euler point vortices, or the equivalent 2-dimensional Coulomb gas,
are a classical topic in Statistical Mechanics, and a well established literature is devoted to them.
The present contribution to such theory consists in determining the rate at which correlations of
vortices, \emph{i.e.} charges, decay in the Mean Field limit.

We will consider the 2-dimensional torus $\T^2=(\R/2\pi\Z)^2$ as space domain;
other 2-dimensional compact manifolds without boundary, or bounded domains of $\R^2$ with smooth boundaries
can be covered by minor modifications of our arguments.
The Euler point vortices system on $\T^2$ consists in $N$ point particles at positions $x_i\in \T^2$ and intensities $\xi_i\in\R$,
satisfying the system of ordinary differential equations
\begin{equation*}
\dot x_{i,t}=\sum_{j\neq i} \xi_j \nabla^\perp G(x_{i,t},x_{j,t}),
\end{equation*}
where the interacting potential is given in terms of the zero-averaged Green function of the Laplace operator,
$\Delta G(x,y)=\delta_x(y)-1$, and $\nabla^\perp=(-\de_2,\de_1)$.
The system is defined so that the vorticity distribution $\omega=\sum \xi_i \delta_{x_i}$ solves the 2-dimensional Euler equations in weak sense, 
see \cite{marchioropulvirenti,Sc96}. It is a Hamiltonian system with respect to the conjugate coordinates $(\xi_i x_{i,1},x_{i,2})$, and Hamiltonian function
\begin{equation*}
H(x_1,\dots,x_n)=\sum_{i< j}^N \xi_i\xi_j G(x_i,x_j),
\end{equation*}
that is the interaction energy of the vortices. The time evolution, which is only defined for a full Lebesgue-measure set
of initial conditions, \cite{DuPu82,marchioropulvirenti,Gr19}, preserves the \emph{canonical Gibbs ensemble},
\begin{equation*}
\nu_{\beta,N}(dx_1,\dots,dx_n)= \frac{1}{Z_{\beta,N}} \exp\pa{-\beta H(x_1,\dots,x_n)}dx_1,\dots,dx_n.
\end{equation*}
This measure was first introduced by Onsager in this context, \cite{onsager}.
Equilibrium ensembles at high kinetic energy, which exhibit the tendency to cluster vortices 
of same sign intensities expected in a turbulent regime,
were proposed by Onsager allowing negative values of $\beta$, a parameter which does \emph{not} correspond to 
the inverse temperature of the fluid.
Unfortunately, we will not be able to treat the case $\beta<0$ with our arguments.

On the torus $\T^2$, in the Mean Field scaling limit, that is in the limit $N\rightarrow\infty$,
$\beta\rightarrow 0$, $N\beta=1$, the $k$-particle correlation function of the Gibbsian enseble
converge to $1$. In other words, in such limit the positions of vortices completely decorrelates. 
To evaluate the rate at which this happens we will resort to Gaussian integration,
transforming functionals of the Gibbsian ensemble into Gaussian integrals,
in fact exploiting techniques dating back to classical works on statistical mechanics
of the Coulomb Gas, such as \cite{Fr76,Sa78,Br78,BrFe80,Ke83,Ke84}.

\subsection{Outline and Notation}

The forthcoming section reviews the classical Mean Field theory of point vortices and states our main result,
\autoref{thm:mainresult}, whereas \autoref{sec:coulomb} contains a mostly formal exposition of the equivalence between
point vortices or 2D Coulomb gas ensembles and the Sine-Gordon Euclidean field theory.
Such correspondence is in fact the origin of the ideas in the proof of \autoref{thm:mainresult},
which is the object of \autoref{sec:decay}.

Throughout the paper, the symbols $\simeq, \lesssim$ denote (in)equalities up to uniform multiplicative factors.
The symbol $\sim$ denotes equality in law of random variables.
The letter $C$ denotes possibly different constants, depending only on its eventual subscripts. Finally, $\chi_A$
is the indicator function of the set $A$. 

\section{Mean Field Theory and Previous Results}

Our discussion begins with a brief review of the Mean Field theory for point vortices on the torus $\T^2$.
We consider a system of an even number $N$ of vortices with positions
\begin{equation*}
	\pa{x_1,\dots x_{N}}=\pa{y_1,\dots y_{N/2}, z_1,\dots z_{N/2}};
\end{equation*}
the first $N/2$ vortices have intensity $+1$, the others $-1$. For brevity, we will denote
$\ubar x=(\ubar y,\ubar z)\in\T^{2\times N}$ the array of all positions.
We consider the Canonical Gibbs measure at inverse temperature $\beta$ associated 
to the Hamiltonian
\begin{equation}\label{eq:vortexhamiltonian}
	H_{N}(\ubar x) = \frac12\sum_{i\neq j}^{N/2} \pa{G(y_i,y_j)+G(z_i,z_j)}-\sum_{i=1}^{N/2}\sum_{j=1}^{N/2} G(y_i,z_j).
\end{equation}
In order to avoid redundant notation, we already introduce in the definition of Gibbs' measures
the Mean Field Limit scaling, $\beta\mapsto \frac\beta{N}$.

\begin{lem}[\cite{deutschlavaud,Fr76,gunsonpanta}]
	For any $0\leq \beta< 4\pi N$,
	\begin{equation*}
		Z_{\beta,N}=\int_{\T^{2\times N}} e^{-\frac{\beta}{N}H_{N}(\ubar x)}dx^{N}<\infty,\quad
		d\mu_{\beta,N}(\ubar x)=\frac1{Z_{\beta,N}} e^{-\frac{\beta}{N}H_{N}(\ubar x)}dx^{N},
	\end{equation*}
	defines a probability measure on $\T^{2\times N}$, symmetric in its first $N/2$ variables $y_i$
	and in the second $N/2$ variables $z_i$.
\end{lem}

The central object of our discussion is the $k$-point correlation function,
the aim being understanding its asymptotic behaviour in the limit $N\rightarrow \infty$.
We fix a finite number of vortices: by symmetry, there is no loss in considering
$(y_1,\dots,y_h,z_1,\dots z_\ell)$ for $N\geq h+\ell$. To ease notation, we will write
\begin{equation*}
	\ubar x=(\hat x,\check{x}), \quad \hat x=(\hat y,\hat z)=(y_1,\dots,y_h,z_1,\dots z_\ell),
\end{equation*}
and analogously $\check{x}$ the array of vortices we are not fixing.
We define
\begin{equation*}
	\rho_{h,\ell}^N(y_1,\dots,y_h,z_1,\dots z_\ell)
	=\rho_{h,\ell}^N(\hat x)=\frac1{Z_{\beta,N}}\int_{\T^{N-h-\ell}} e^{-\frac\beta{N}H_{N}(\ubar x)}d\check{x}.
\end{equation*}
Here and from now on $d\check x$ (respectively $d\hat x$) indicates integration with respect
to the $N-h-\ell$ 2-dimensional variables $\check{x}$ (resp. the $h+\ell$ variables $\hat{x}$).

\begin{thm}\label{thm:lions}
	Let $\beta>0$; the free energy functional
	\begin{align}\label{eq:freeenergy}
		&\F(\rho_+,\rho_-)=\frac1\beta\int_{\T^2}(\rho_+\log\rho_++\rho_-\log\rho_-)
		+\int_{\T^2} (\rho_+-\rho_-)G\ast(\rho_+-\rho_-),\\ \nonumber
		&\quad \rho_+,\rho_- \text{ probability densities on }\T^2\text{ such that }\rho_\pm\log\rho_\pm\in L^1(\T^2),
	\end{align}
	admits the unique minimiser $\rho_+=\rho_-\equiv 1$.
	For any $1\leq h+\ell\leq N$ and $1\leq p<\infty$, the $(h+\ell)$-point correlation function $\rho^N_{h,\ell}$
	converges to $\rho_+^{\otimes h}\otimes \rho_-^{\otimes \ell}\equiv 1$ in $L^p$ topology, 
	\begin{equation}\label{eq:qualitativedecay}
		\lim_{N\rightarrow\infty} \norm{\rho^N_{h,\ell}-1}_{L^p(\T^{2\times N})}=0.
	\end{equation}
\end{thm}

The latter is a classical result, valid for more general geometries of the space domain and for small negative temperatures
regimes, although in such generality the minimiser of the functional (maximiser for $\beta<0$) might not be unique and
limit points of the sequence $(\rho^N_{k,h})_{N\in\N}$ can thus be superpositions of minima (resp. maxima) of $\F$.
We refer to \cite{clmp92,clmp95} and the monography \cite{lionsbook} for a complete discussion.

Stationary points of the free energy can be characterised as solutions of the Mean Field equation
for the potential $\phi=G\ast(\rho_+-\rho_-)$,
\begin{equation*}
	-\Delta\phi=\frac{e^{-\beta\phi}}{Z_+}-\frac{e^{\beta\phi}}{Z_-}, \quad Z_\pm=\int_{\T^{2}}e^{\mp \beta\phi}dx,
\end{equation*}
which, up to a suitable choice of the average $\psi=\phi+c$, is equivalent to the sinh-Poisson equation,
\begin{equation}\label{eq:sinhpoisson}
	\Delta\psi=\frac1\alpha \sinh(\beta\psi), \quad 4\alpha^2=\int_{\T^{2}}e^{-\beta\psi}dx\int_{\T^{2}}e^{\beta\psi}dx,
\end{equation}
see \cite[section 7.5]{marchioropulvirenti}.
Since on the torus there is a unique and trivial solution $\rho\equiv 1$,
such equivalence is trivial in our setting: it is nonetheless a more general fact.

The main result of the present paper is the following refinement of \autoref{thm:lions},
concerning the rate at which the convergence \eqref{eq:qualitativedecay} takes place.

\begin{thm}\label{thm:mainresult}
	For any $\beta>0$, $1\leq k+h\leq N$ and $1\leq p<\infty$,
	\begin{equation*}
		\|\rho^N_{h,\ell} - 1\|_{L^p(\T^{2\times N})}\leq \frac{C_{\beta,p,h,\ell}}{\sqrt N}(\log N)^{\frac32}.
	\end{equation*}
\end{thm}

The core idea behind our computations is the correspondence, provided by Gaussian integration,
between functionals of the vortex ensemble and certain Euclidean field theoretic integrals.
We are able to exploit such link, to be outlined in the forthcoming section, only
for positive temperatures, $\beta>0$. This unfortunately rules out a relevant regime, $\beta<0$,
in which the Mean Field equation on $\T^2$ admits nontrivial solutions, see \cite{lionsbook}.

\section{The Coulomb Gas and Sine-Gordon Field Theory}\label{sec:coulomb}

The 2-dimensional, Coulomb gas is a classical mechanics system
consisting of point charges: we will consider the case in which there are two species of charges
of opposite signs, but with same intensity. For a system of $N$ charges,
say half positive and half negative, their dynamics is described
by the Hamiltonian function
\begin{equation*}
	\H=\frac12 \sum_{i=1}^N p_i^2+\frac12 \sum_{i\neq j}\sigma_i\sigma_j G(x_i,x_j),
\end{equation*}
where $G$ is the Green function of the Laplacian, as above, $x_i$ are the positions and $p_i$
the momenta of the charges, $\sigma_i=\pm 1$ the signs of the charges.
In Gibbsian ensembles of the system, momenta have Maxwellian (Gaussian) independent distributions;
when dealing with correlation functions or analogous functionals --which is ultimately the
aim of the present work-- we can always integrate them out: it is thus convenient to only consider the configurational
(interaction) part of the Hamiltonian.

We consider the system of charges in a bounded domain $D\subseteq \R^2$, 
so boundary conditions have to be supplemented to define $G$:
for the sake of this discussion there is no difference in considering free boundary conditions,
Dirichlet boundary conditions (physically interpreted as considering the system in a cavity inside a conductor)
or the periodic case $\T^2$.
It is immediate to observe that the (configurational) Canonical Gibbs ensemble for the 2D Coulomb gas
actually coincides with the vortices ensemble defined above, provided that the same boundary conditions are taken into account,
since the configurational part of the Hamiltonian $\H$ is in fact the same as \eqref{eq:vortexhamiltonian}.

\subsection{The Sine-Gordon representation}

It is a classical and well-known fact that Gaussian integration provides a
correspondence between two-dimensional Coulomb gas and the Sine-Gordon field theory,
as described in \cite{Sa78}.
This equivalence has been instrumental in the study of both systems, see for instance \cite{Fr76,Br78,BrFe80},
since it allowed to employ techniques from both statistical mechanics and field theory.
The remainder of this section is dedicated to review such correspondence,
which we will exploit in the proof of \autoref{thm:mainresult}. The following arguments are mostly formal and not rigorous:
indeed we only aim to provide a heuristic motivation of the techniques we are going to use.

The equivalence with Sine-Gordon theory is exact only when the Coulomb gas is considered in the Grand Canonical ensemble.
Let us then consider the (configurational part of the) Grand Canonical partition function,
\begin{align}\label{grandcanonical}
\mathfrak{Z}_{z,\beta}
&=\sum_{n=0}^{\infty} \frac{z^n}{n!}
\int_{D^n}\exp\pa{-\beta H_n(x_1,\sigma_1,\dots x_n,\sigma_n)}dx^n d\nu^n,
\end{align}
where the \emph{activity} $z>0$ controls the arbitrary (Poisson distributed) number $n$ of charges
and $\nu$ is the law of a $\frac12$-Bernoulli variable on $\set{\pm1}$;
the positions $x_i$ and signs $\sigma_i$ are thus independent variables with law, respectively, $dx$ on $\T^2$ and $\nu$.
Notice that the neutrality condition has been replaced with an average neutrality, $\int\sigma d\nu(\sigma)=0$;
this is only for the sake of simplicity of exposition, different and more general choices can be made.

The corresponding (Euclidean) Sine-Gordon field theory has Lagrangian
\begin{equation*}
	\L(\phi)=\beta\abs{\nabla\phi}^2 -2z \cos\pa{\beta \phi},
\end{equation*}
so that the vacuum expectation value is
\begin{equation*}
	\V_{z,\beta}= \int e^{-\int \L(\phi)dx}\D\phi=
	\int \exp\pa{-\beta\int_D |\nabla\phi|^2dx +2z\int_D \cos(\beta \phi)dx}\D\phi.
\end{equation*}
The equivalence with Grand Canonical Coulomb gas is most immediately seen by observing that the partition function $\mathfrak{Z}_{z,\beta}$
actually coincides with the Sine-Gordon vacuum expectation, up to a normalising factor given by the vacuum expectation of the free field,
\begin{equation}\label{eq:sinegordon}
	\mathfrak{Z}_{z,\beta}=\V_{z,\beta}/\V_{0,\beta}.
\end{equation}
This can be shown with the following formal computation. If $X,Y$ are two real standard Gaussian variables, it holds
\begin{equation*}
	e^{\frac{s^2+t^2}{2}}\expt{e^{\imm s X} e^{\imm t Y}}=e^{-st\expt{XY}}.
\end{equation*}
By means of this Fourier transform, we can thus formally see any exponential function $e^{-G(x_i,x_j)}$ as the 
field theoretic correlation function of the field operators $e^{\imm\chi(x_i)},e^{\imm\chi(x_j)}$
with respect to the free (Gaussian) theory with action $\int|\nabla \chi|^2dx$ (the 2-dimensional \emph{Gaussian free field}).
More explicitly, we write
\begin{equation*}
	\frac{\int e^{\imm \beta \sum_{i=1}^n \sigma_i \chi(x_i)} e^{-\beta\int_D |\nabla\phi|^2dx}\D\phi}{\int e^{-\beta\int_D |\nabla\phi|^2dx}\D\phi}
	=\exp\pa{-\frac{\beta}{2} \sum_{i\neq j}^n \sigma_i \sigma_j G(x_i,x_j)}.
\end{equation*}
The computation is only formal since the random field $\chi$ has singular covariance: its samples are not functions
($\chi$ can be realised as a random distribution), and thus the above complex exponentials need renormalisation to be
rigorously defined. Proceeding with the formal computation 
(in which for a moment we omit the infinite renormalisation term $\V_{0,\beta}=\int e^{-\beta\int |\nabla\phi|^2}\D\phi$),
\begin{align*}
	&\sum_{n=0}^{\infty} \frac{z^n}{n!} \int dx_1\cdots dx_n d\nu(\sigma_1)\cdots d\nu(\sigma_n)
	\int e^{-\beta\int_D |\nabla\phi|^2dx}\D\phi e^{\imm \beta \sum_{i=1}^n \sigma_i \chi(x_i)} \\
	&\quad = \sum_{n=0}^{\infty} \frac{z^n}{n!} \int e^{-\beta\int_D |\nabla\phi|^2dx}\D\phi 
	\pa{\int dxd\nu(\sigma) e^{\imm \beta \sigma \chi(x)}}^n\\
	&\quad =\int e^{-\beta\int_D |\nabla\phi|^2dx}\D\phi \sum_{n=0}^{\infty} \frac{z^n}{n!} \pa{\int dx 2\cos(\beta \chi(x))}^n\\
	&\quad =\int e^{2z\int \cos(\beta \phi(x))dx} e^{-\beta\int_D |\nabla\phi|^2dx}\D\phi=\V_{z,\beta}, 		
\end{align*}
from which \eqref{eq:sinegordon}.

\subsection{Mean Field Scaling and Correlation Functions}

The Mean Field scaling of Coulomb charges in the Canonical ensemble is
\begin{equation*}
	\beta\mapsto \epsilon \beta, \quad N\mapsto \frac{N}{\epsilon}, \quad \epsilon\rightarrow 0,
\end{equation*}
and it corresponds in the Grand Canonical Ensemble to
\begin{equation*}
    \beta\mapsto \epsilon \beta, \quad z\mapsto \frac{z}{\epsilon}, \quad \epsilon\rightarrow 0
\end{equation*}
($\epsilon$ sometimes referred to as the \emph{plasma parameter}).
Applying the Mean Field scaling to the Sine-Gordon theory one recovers the Klein-Gordon field theory:
looking at vacuum expectations,
\begin{equation*}
	\V_{z/\epsilon,\epsilon\beta}\xrightarrow{\epsilon\rightarrow 0}
	\int \exp\pa{\int_D |\nabla\phi|^2dx +z\beta \int_D \phi^2 dx}\D\phi,
\end{equation*}
the right-hand side being the vacuum expectation of the theory with Lagrangian
\begin{equation*}
	\L(\phi)=\abs{\nabla\phi}^2 -z\beta \phi^2,
\end{equation*}

This is because in such a scaling every term in the power expansion of the interaction term $\cos(\xi \sqrt\beta \phi)$
is negligible save for the quadratic one. A straightforward computation 
--using for instance Fourier series on $\T^2$-- reveals that the Mean Field scaling limit of $\mathfrak{Z}_{z,\beta}$
in fact coincides with the partition function of the Energy-Enstrophy invariant measure of the 2-dimensional Euler
equations,
\begin{align*}
	\mathfrak{Z}_{z/\epsilon,\epsilon\beta}&=\V_{z/\epsilon,\epsilon\beta}/\V_{0,0}\xrightarrow{\epsilon\rightarrow 0}
	Z_\beta=
	\int \exp \pa{-\beta\int_D \omega\Delta^{-1}\omega dx}d\mu(\omega),\\
	d\mu(\omega)&=\frac1Z \int e^{-\int_D \omega^2 dx}\D\omega, 
\end{align*}
where $\mu$ --the Enstrophy measure-- is actually the space white noise on $\T^2$.
The following result of \cite{GrRo19} (to which we refer for a complete discussion of the involved Gaussian measures),
rigorously establishes such convergence for the Canonical ensemble of charges on the torus $\T^2$.

\begin{thm}
	For any $\beta\geq 0$, 
	\begin{equation*}
		\lim_{N\rightarrow\infty} Z_{\beta,N}=Z_\beta.
	\end{equation*}
\end{thm}

Let us now fix the first $k$ charges, with positions $x_1,\dots x_k\in D$ and intensities 
$\xi_i=\sigma_i$, $\sigma_i\in \set{\pm 1}$, $i=1,\dots k$.
Their Grand Canonical correlation function is obtained considering the ensemble composed of those and other $n$ charges
with random position and intensities, $n$ being also randomly distributed as before,
\begin{align*}
	\rho(x_1,\xi_1,\dots x_k,\xi_k)=\frac1{\mathfrak{Z}_{z,\beta}}	
	\sum_{n=1}^{\infty} \frac{z^n}{n!}
	\int_{D^n}e^{-\beta H_{n+k}(x_i,\sigma_i)} \prod_{i=k+1}^{n+k} dx_{i}d\nu(\sigma_i) 
\end{align*}
In the Sine-Gordon correspondence, these statistical mechanics correlation functions transform into
the correlation (Green function) of the field operators $e^{\imm \xi_i \chi(x_i)}$,
\begin{equation}\label{eq:corrsingor}
	\rho(x_1,\xi_1,\dots x_k,\xi_k)=
	\frac{\int \prod_{i=1}^k e^{\imm \sqrt\beta \sigma_i \phi(x_i)} e^{-\frac1\epsilon\int_D \L(\phi) dx}\D\phi}
	{\int e^{-\frac1\epsilon\int_D \L(\phi) dx}\D\phi}
\end{equation}
The latter expression follows from the same formal computations of the previous paragraph:
we applied the Gaussian integration formula with respect to the free field with Lagrangian $\frac1\epsilon\int|\nabla\phi|^2dx$,
so that the dependence on $\epsilon$ is factored out from the action.

As $\epsilon$ goes to zero, the dominant contribution of the functional integrals in \eqref{eq:corrsingor}
comes from the stationary points of the action $\S(\phi)=\int_D \L(\phi)dx$, which are given by
\begin{equation*}
	\frac{\delta\S}{\delta \phi}=\Delta \phi-2z\sin(\sqrt\beta\phi)=0,
\end{equation*}
which is equivalent, setting $\psi=-\imm \phi$, to the Debye-H\"uckel Mean Field equation,
\begin{equation*}
	\Delta\psi=2z\sinh(\sqrt\beta\psi),
\end{equation*}
which is a sinh-Poisson equation in agreement with the one in \eqref{eq:sinhpoisson}.
In the particular case of the torus, $D=\T^2$, this equation only admits the trivial solution $\psi\equiv 0$.
The limit of correlation functions can thus be obtained by evaluating the field operator $\prod_{i=1}^k e^{\imm \sqrt\beta \sigma_i \phi(x_i)}$
at the stationary point,
\begin{equation*}
	\rho(x_1,\xi_1,\dots x_k,\xi_k)\sim \prod_{i=1}^k e^{-\sqrt\beta \xi_i \psi(x_i)}=1.
\end{equation*}
Formal computations involving power expansion of the cosine interaction term leads to further orders
behaviour of the correlation function in $\epsilon$, see \cite{Ke83}. 

Our work actually finds an analogue in \cite{Ke83}, with some important differences: 
they consider Coulomb charges in dimension 3 (while we exclusively focus on the 2-dimensional case),
and their charges are smeared, the cutoff parameter going to zero in a suitable rate with respect to the Mean Field scaling,
while we retain the whole singularity of the interaction. The latter difference is analogous to
the one between the two works \cite{BePiPu87} and \cite{GrRo19}.

\section{Decay of Correlations}\label{sec:decay}

Let us now proceed to the proof of our main result, \autoref{thm:mainresult}.
The main difficulty is due to the logarithmic singularity of the Green function $G$. 
To deal with it we will decompose $G$ in two parts, a smooth approximation $V_m$ of $G$ 
and a remainder $W_m$ retaining logarithmic singularity.
Thanks to $V_m$ being smooth, we can apply the Sine-Gordon transformation to the associated part of the
Gibbs exponential, and then obtain the sought asymptotic behaviour by means of an iterative expansion
of Gaussian exponentials.
The contribution of the singular part $W_m$ turns out to be negligible when we consider the limit in
$N\rightarrow\infty$ with $m=m(N)\rightarrow \infty$ in a suitable rate.

\subsection{Potential Splitting and Preliminary Results}

First and foremost, let us split the interaction potential $G$: for $m>0$,
\begin{equation}\label{potentialsplitting}
G=-\Delta^{-1}=\pa{-\Delta^{-1}-(m^2-\Delta)^{-1}}+(m^2-\Delta)^{-1}:= V_m+W_m.
\end{equation}
Physically, the singular part $W_m$ corresponds to the singular short-range part of the interaction:
indeed the Green function of $m^2-\Delta$ with free boundary conditions --the 2-dimensional \emph{Yukawa potential}
or \emph{screened Coulomb potential} with mass $m$--
has logarithmic divergence in the origin but decays exponentially fast at infinity.

We will denote, according to \eqref{potentialsplitting},
\begin{equation*}
H=H_{V_m}+H_{W_m}=\sum_{i< j}^N \xi_i\xi_j V_m(x_i,x_j) + \sum_{i< j}^N \xi_i\xi_j W_m(x_i,x_j).
\end{equation*}
We will regard the regular part of the Hamiltonian corresponding to $V_m$ as the covariance
of a Gaussian field as we formally did in \autoref{sec:coulomb} for the full Hamiltonian. 
We thus define $F_m$ as the centred Gaussian field on $\T^2$ with covariance kernel $V_m$, that is
\begin{equation}
\forall f,g\in \dot L^2(\T^2), \quad 	
\expt{\brak{F_m,f}\brak{F_m,g}}=\brak{f,\pa{-\Delta^{-1}-(m^2-\Delta)^{-1}} g}.
\end{equation}
The remainder of this paragraph deals with properties of $F_m$. The reproducing kernel Hilbert space is
\begin{equation*}
\sqrt{-\Delta^{-1}-(m^2-\Delta)^{-1}}\dot L^2(\T^2)\subseteq \dot H^2(\T^2),
\end{equation*} 
so that $F_m$ has a $\dot H^s(\T^2)$-valued version for all $s<1$, into which $\dot H^2(\T^2)$ has Hilbert-Schmidt embedding.
As a consequence, by Sobolev embedding, $F_m$ has a version taking values in $\dot L^p(\T^2)$ for all $p\geq 1$.

The field $F_m$ can also be evaluated at points $x\in\T^2$: the coupling $F_m(x):=\brak{\delta_x,F_m}$ is
defined as the series, converging in $L^2(F_m)$ \emph{uniformly in} $x\in\T^2$,
\begin{equation*}
\brak{\delta_x,F_m}=\sum_{k\in\Z^2_0} e^{2\pi \imm x\cdot k} \hat F_{m,k},
\quad \hat F_{m,k}=\brak{e_k,F_m}\sim N_\C\pa{0,\frac{m^2}{4\pi^2 |k|^2\pa{m^2+4\pi^2|k|^2}}}.
\end{equation*}
In other terms, $x\mapsto F_m(x)$ is a measurable random field, and 
$F_m(x)$ are centred Gaussian variables of variance $V_m(x,x)=V_m(0,0)$.
By Kolmogorov continuity theorem, there also exists a version of $F_m(x)$
which is $\alpha$-H\"older for all $\alpha<1/2$.

\begin{lem}\label{lem:Fbounds}
	For any $\alpha>0$, $p\geq 1$ and $m\rightarrow \infty$,
	\begin{align}
	\label{Fmoments}
	\expt{\norm{F_m}_p^p}&\simeq_p (\log m)^{p/2},\\
	\label{Fexpmoments}
	\expt{\exp\pa{-\alpha\norm{F_m}_2^2}}&\simeq m^{-\frac{\alpha}{2\pi}},
	\end{align}
	and, moreover, for $0<\alpha\leq\alpha'$,
	\begin{equation}\label{Fexpmomentsdiff}
		\expt{\exp(-\alpha\|F_m\|_{L^2}^2)} - \expt{\exp(-\alpha'\|F_m\|_{L^2}^2)}
		\lesssim (\alpha' - \alpha)m^{-\frac\alpha{2\pi}}\log m.
	\end{equation}
\end{lem}

\begin{proof}
	Let us begin with \eqref{Fmoments}: by Fubini-Tonelli,
	\begin{equation*}
	\expt{\norm{F_m}_p^p}=\int_{\T^{2}}\expt{|F_m(x)|^p}dx=c_p \int_{\T^{2}}V_m(x,x)^{p/2}dx=c_p V_m(0,0)^{p/2},
	\end{equation*}
	where $V_m(0,0)=\frac{1}{2\pi}\log m+o(\log m)$ can be checked by explicit computation in Fourier series.
	As for \eqref{Fexpmoments}, a standard Gaussian computation (see \cite[Proposition 2.17]{dpz}) gives
	\begin{align*}
	\expt{\exp\pa{-\alpha \norm{F_m}_2^2}}
	&=\exp \set{-\frac{1}{2}\trace\pa{\log\pa{1+2\alpha \pa{-\Delta^{-1}-(m^2-\Delta)^{-1}}}}}\\
	&=\exp\pa{-\frac{1}{2}\sum_{k\in\Z^2_0}\log \pa{1+\frac{2\alpha m^2}{4\pi^2 |k|^2 (m^2+4\pi^2|k|^2)}}}\\
	&>\exp \pa{-\sum_{k\in\Z^2_0}\frac{\alpha m^2}{4\pi^2 |k|^2 (m^2+4\pi^2|k|^2)}}\\
	&=\exp\pa{-\alpha V_m(0,0)}\simeq m^{-\frac{\alpha}{2\pi}},  
	\end{align*}
	the other inequality descending from analogous computations using $\log(1+x)>x-\frac{x^2}{2}$, $x>0$, instead of
	the inequality $\log(1+x)<x$ we just applied. 
	Finally, \eqref{Fexpmomentsdiff} is obtained considering the first order Taylor expansion of the exponential
	and controlling the remainder by means of Gaussian computations analogous to the ones above.
\end{proof}

Let us now consider the Sine-Gordon transformation applied to $H_{V_m}$:
since for any $s,t\in\R$ it holds
\begin{equation*}
\expt{e^{\imm s F_m(x)} e^{\imm t F_m(y)}}=e^{-\frac{s^2+t^2}{2}V_m(0,0)}e^{-stV_m(x,y)},
\end{equation*}
(and analogous expressions for $n$-fold products) we can transform
\begin{align}\label{sinegordontorus}
\int_{\T^{2N}} & e^{-\beta H_{V_m}} dx_1\cdots dx_n\\ \nonumber
&= \int_{\T^{2N}} \exp\pa{-\frac\beta{2N}\sum_{i\neq j}^{N}\sigma_i\sigma_j V_m(x_i,x_j)}dx_1\cdots dx_n\\ \nonumber
&=e^{\frac{\beta}{2}V_m(0,0)}
\expt{\int_{\T^{2N}} \exp\pa{-\imm \sqrt{\frac{\beta}{N}}\sum_{i=1}^{N} \sigma_i F_m(x_i)}dx_1\cdots dx_n}.
\end{align}
In both expressions, $\mathbb{E}$ denotes expectation with respect to the law of the Gaussian field $F_m$.
It is worth recalling the following estimate on the regular Gibbs partition function obtained in \cite[Proposition 2.8]{GrRo19}
(see also \autoref{prop:Fextrabounds} below).

\begin{prop}\label{prop:Fcomplexbounds}
	For any $\beta>0$ and integer $n\geq 1$, if $m=m(N)$ grows at most polynomially in $N$, then it holds
	\begin{equation*}
	\int_{\T^{2N}} e^{-\beta H_{V_m}} dx_1\cdots dx_n
	\leq C_{\beta,n} \pa{1+\frac{m^{\frac{\beta}{4\pi}}\pa{\log m}^{2n}}{N^{n/2}}}
	\end{equation*}
	uniformly in $N$.	
\end{prop}

As already remarked, we will also need some control on (the partition function associated to) the singular part of the
potential $W_m$. We refer to \cite[Proposition 2.9]{GrRo19} for the proof of the following:

\begin{prop}\label{prop:zyukawatorus}
	Let $N\geq 1$, $\beta> -8\pi$ and $m>0$. 
	There exists a constant $C_{\beta}>0$ such that
	\begin{equation*}
	\int_{\T^{2N}} e^{-\beta H_{W_m}}dx_1\cdots dx_n \leq \pa{1+C_{\beta}\frac{(\log m)^2}{m^2}}^N.
	\end{equation*}
\end{prop}

Finally, we will need some elementary properties of real and complex exponential integrals,
which we isolate here for the reader's convenience.

\begin{lem}\label{lem:exponentialintegral}
	Let $(X,\mu)$ be a probability space and $f\in L^1(X,\mu)$ 
	with $\int f d\mu=0$ and $\int e^{-\alpha f}d\mu<\infty$ for $\alpha>0$.
	Then for all $n\geq1$,
	\begin{equation*}
		\int \pa{e^{-f}-1}^{2n} d\mu \leq 2^{2n-2} \int (e^{-2nf}-1)d\mu.
	\end{equation*}
	Moreover, if additionally $f\in L^4(X,\mu)$ , then
	\begin{equation*}
	\abs{\int e^{\imm f} d\mu-e^{-\frac{1}{2}\norm{f}^2_{2}}}\leq \frac{\norm{f}^3_3}{6}+\frac{\norm{f}^4_2}{8}.
	\end{equation*} 
\end{lem}
\begin{proof}
	Expanding the product,
	\begin{equation*}
		\int \pa{e^{-f}-1}^{2n} d\mu = \sum_{k=0}^{2n}\binom{2n}{k}(-1)^k\int e^{-kf}d\mu,
	\end{equation*}
	and controlling positive and negative terms respectively with Young's and Jensen's inequalities,
	\begin{equation*}
		1\leq e^{-k\int fd\mu}\leq \int e^{-kf}d\mu \leq \frac{k}{2n}\int e^{-2nf} d\mu+ \frac{2n-k}{2n},
	\end{equation*}
	we get
	\begin{align*}
		\int \pa{e^{-f}-1}^{2n} d\mu
		&\leq \pa{\sum_{k=0}^n\frac{k}{n}\binom{2n}{2k}}\int e^{-2nf} d\mu\\
		&\qquad + \sum_{k=0}^n\binom{2n}{2k}\frac{n-k}{n}
		- \sum_{k=0}^{n-1}\binom{2n}{2k+1}\\
		&= 2^{2n-2}\int (e^{-2nf} - 1)d\mu,
	\end{align*}
	which proves the first statement.
	As for the second one, thanks to the zero average condition, we can expand
	\begin{align*}
	&\int_{\T^2} e^{\imm f(x)} dx-e^{-\frac{1}{2}\norm{f}^2_{2}}\\
	& \quad = \int_{\T^2} \pa{e^{i f(x)}-1-\imm f(x)+\frac{f(x)^2}{2}} dx
	-\pa{e^{-\frac{1}{2}\norm{f}^2_{2}}-1+\frac{\norm{f}^2_2}{2}}
	\end{align*}
	and then apply Taylor expansions
	\begin{equation*}
	\abs{e^{it}-1-it+\frac{t^2}{2}}\leq \frac{t^3}{6}, \quad \abs{e^{-t}-1+t}\leq\frac{t^2}{2}. \qedhere
	\end{equation*}
\end{proof}

\subsection{Proof of \autoref{thm:mainresult}}

To ease notation, in the following argument we will denote
\begin{equation*}
E_j=\int_{\mathbb{T}_2}e^{\imm\xi_j\sqrt{\beta}F_m(x_j)}\,dx_j,
\qquad \Es =e^{-\frac{\beta}{2N\gamma}\|F_m\|_{L^2}^2},
\end{equation*}
(notice that both depend on $N,m=m(N)$) and thus write \eqref{sinegordontorus} as
\begin{equation*}
\int_{\T^{2N}} e^{-\beta H_{V_m}} dx_1\cdots dx_n =e^{\frac{\beta}{2\gamma}V_m(0,0)}
\expt{\prod_{j=1}^N E_j}
\end{equation*}	
In sight of \autoref{lem:exponentialintegral}, we expect the $0$-th order term (in $1/N$)
to be given by $e^{\frac{\beta}{2\gamma}V_m(0,0)}\expt{\Es^N}$, which is $O(1)$ as shown above in \autoref{lem:Fbounds}.
The forthcoming proof applies the Taylor expansion of \autoref{lem:exponentialintegral} to further
and further orders.

\begin{prop}\label{prop:Fextrabounds}
	For any $\beta\geq0$ and integer $k\geq 0$, let
	\begin{equation*}
		\mathscr{R}_k= \Bigl(\prod_{j=k+1}^N E_j\Bigr) - \Es^{N-k}.
	\end{equation*}
	If $m=m(N)$ grows at most polynomially in $N$, for every integer $n\geq1$
	\begin{equation*}
		\mathbb{E}[|\mathscr{R}_k|]
		\leq \frac{C_{\beta,k,n}}{\sqrt N}m^{-\frac\beta{4\pi}}(\log m)^{\frac32}
		+ \frac{C_{\beta,k,n}}{N^{\frac{n}2}}(\log m)^{3n/2}.
	\end{equation*}
\end{prop}
\begin{proof}
	For $n=1$, we expand the product $\prod_{j=k+1}^N E_j$ by means of the algebraic
	identity 
	\begin{equation}\label{algfirstorder}
	\prod_{j=k+1}^N E_j
	= \Es^{N-k}+\sum_{\ell=k+1}^{N} (E_\ell-\Es)\Es^{N-\ell}
	\Biggl(\prod_{j=k+1}^{\ell-1} E_j\Biggr).
	\end{equation}
	For $n=2$, by iterating \eqref{algfirstorder} we get the identity
	\[
	\begin{aligned}
	\mathscr{R}_k
	&= \Es^{N-k-1}\sum_{\ell=k+1}^N(E_{\ell}-\Es)\\
	&\quad  + \sum_{k+1\leq\ell_1<\ell_2\leq N}^N
	\Es^{N-k-\ell_1+1}
	(E_{\ell_1}-\Es)(E_{\ell_2}-\Es)
	\Biggl(\prod_{j=k+1}^{\ell_1-1}E_j\Biggr).
	\end{aligned}
	\]
	For general $n$, the iteration of \eqref{algfirstorder}
	yields,
	\[
	\begin{aligned}
	\mathscr{R}_k
	&= \sum_{\ell=1}^{n-1}\Es^{N-k-\ell}
	\sum_{k+1\leq k_1<\dots< k_\ell\leq N}\prod_{j=1}^\ell(E_{k_j}-\Es)\\
	&\quad +\sum_{k+1\leq k_1<\dots< k_n\leq N}
	\Es^{N-n-k_1+1}
	\Biggl(\prod_{j=1}^n(E_{k_j}-\Es)\Biggr)
	\Biggl(\prod_{j=k+1}^{k_1-1}E_j\Biggr).
	\end{aligned}
	\]
	To estimate the expectation of $\mathscr{R}_k$, everything
	boils down to estimate expectations of terms $\Es^a\|F_m\|_{L^3}^{3b}$
	for $a,b>0$.
	Indeed, we notice that $|E_j|\leq 1$ and $\Es\leq 1$, and that by Taylor
	expansion, and since $F_m$ has zero average on the torus,
	$|E_j-\Es|\leq N^{-3/2}\|F_m\|_{L^3}^3$. By \autoref{lem:Fbounds}
	and Cauchy-Schwarz,
	\[
	\mathbb{E}[\Es^a\|F_m\|_{L^3}^{3b}]
	\leq \mathbb{E}[\Es^{2a}]^{\frac12}\|F_m\|_{L^3}^{6b}]^{\frac12}
	\leq m^{-\frac{a}{4\pi N}\beta}(\log m)^{3b}.
	\]
	Thus,
	\[
	\begin{aligned}
	\mathbb{E}[|\mathscr{R}_k|]
	&\lesssim \sum_{\ell=1}^{n-1}N^{-\ell/2}
	\mathbb{E}\bigl[\Es^{N-k-\ell}\|F_m\|_{L^3}^{3\ell}\bigr]\\
	&\quad + \frac1N\sum_{k_1=k+1}^{N-n+1} N^{-n/2}
	\mathbb{E}\bigl[\Es^{N-n-k_1+1}\|F_m\|_{L^3}^{3n}\bigr]\\
	&\lesssim \sum_{\ell=1}^{n-1}N^{-\ell/2}
	m^{-\frac{N-k-\ell}{4\pi N}\beta}(\log m)^{3\ell/2}
	+ N^{-n/2}(\log m)^{3n/2}\\
	&\lesssim \frac1{\sqrt N}m^{-\frac\beta{4\pi}}(\log m)^{\frac32}
	+ N^{-n/2}(\log m)^{3n/2},
	\end{aligned}
	\]
	since $m$ is polynomial in $N$, therefore $N^{-1/2}(\log m)^{3/2}m^{\beta/4\pi N}$
	is smaller than $1$ for $N$ large enough.
\end{proof}

\begin{rmk}
	In fact, \autoref{prop:Fextrabounds} reprises the argument used in \cite{GrRo19} to prove \autoref{prop:Fcomplexbounds}:
	indeed, the latter can be deduced from the former.
\end{rmk}

\begin{proof}[Proof of \autoref{thm:mainresult}]
	Fix an even integer $N\geq1$ large enough, an exponent $p\in[1,\infty)$,
	and denote by $q\in(1,\infty]$ the H\"older conjugate exponent, so that
	$1/p+1/q=1$. Let $f\in L^q(\T^{2\times k})$ be a test function such that $\|f\|_{L^q}\leq 1$.
	We use the potential splitting \eqref{potentialsplitting},
	with $m$ polynomial in $N$, to decompose the integral of $f$,
	\begin{align*}
		\int_{\T^{2k}} f(\hat x)\rho^N_{h,\ell}(\hat x)\,d\hat x
		&= \frac1{Z_{\beta,N}}\int_{\T^{2k}} f(\hat x)
		(e^{-\frac\beta{N} H_{W_m}}-1)e^{-\frac\beta{N} H_{V_m}}\,d\hat x\,d\check x\\
		&\quad + \frac1{Z_{\beta,N}}\int_{\T^{2k}} f(\hat x)
		e^{-\frac\beta{N} H_{V_m}}\,d\hat x\,d\check x\\
		&:=[S]+[R].
	\end{align*}
	We first consider $[S]$. Let $r,s\geq1$ be such
	that $1/r+1/s=1/p$, then by the H\"older inequality,
	\begin{equation*}
		[S]	\leq \frac1{Z_{\beta,N}}\|e^{-\frac\beta{N} H_{W_m}}-1\|_{L^r}
		\|e^{-\frac\beta{N} H_{V_m}}\|_{L^s}.
	\end{equation*}
	By Jensen's inequality, $Z_{\beta,N}\geq1$, moreover,
	by \autoref{prop:Fcomplexbounds},
	$\|e^{-\frac\beta{N} H_{V_m}}\|_{L^s}$ is uniformly bounded
	in $N$ by our choice of $m$. If $n$ is
	the smallest integer such that $2n\geq r$ (thus $2n\leq r+2$), by
	\autoref{prop:zyukawatorus} and \autoref{lem:exponentialintegral},
	\begin{equation}\label{eq:rate1}
	\|e^{-\frac\beta{N} H_{W_m}}-1\|_{L^r}
	\leq \Bigl(\int_{\T^{2N}}e^{-\frac{2n\beta}{N} H_{W_m}}-1\Bigr)^{\frac1{2n}}
	\lesssim \Bigl(\frac{N}{m^2}(\log m)^2\Bigr)^{\frac1{r+2}},
	\end{equation}
	since by our choice of $m$, $N/m^2$ converges to $0$ polynomially in $1/N$.
	
	We turn to the estimate of $[R]$. Set
	\[
	\delta(\hat x)
	= \Biggl(\prod_{j=1}^h e^{\imm\sqrt\beta F_m(y_j)}\Biggr)
	\Biggl(\prod_{j=1}^\ell e^{-\imm\sqrt\beta F_m(z_j)}\Biggr),
	\]
	then as in \eqref{sinegordontorus},
	\[
	[R]
	= \frac{1}{Z_{\beta,N}}e^{\frac12\beta V_m(0,0)}
	\mathbb{E}\Biggl[
	\Bigl(\prod_{j=k+1}^N E_j\Bigr)
	\int_{\T^{2k}}f(\hat x)\delta(\hat x)\,d\hat x\Biggr].
	\]
	Consider the two terms that originate from the decomposition
	of the product in $\Es^{N-k}+\mathscr{R}_k$.
	First, by \autoref{prop:Fextrabounds},
	\begin{equation}\label{eq:rate2}
	\begin{aligned}
	\frac{e^{\frac12\beta V_m(0,0)}}{Z_{\beta,N}}
	\mathbb{E}\Bigl[\mathscr{R}_k
	\int_{\T^{2k}}f(\hat x)\delta(\hat x)\,d\hat x\Bigr]
	&\leq \frac{1}{Z_{\beta,N}}e^{\frac12\beta V_m(0,0)}
	\mathbb{E}[|\mathscr{R}_k|]\\
	&\leq\frac{(\log m)^{\frac32}}{\sqrt N}
	+ \frac{m^{\frac\beta{4\pi}}}{N^{n/2}}(\log m)^{3n/2}.
	\end{aligned}
	\end{equation}
	By a Taylor expansion,
	\[
	|\delta(\hat x) - 1|
	\lesssim\frac1{\sqrt N}\sum_{j=1}^h F_m(y_j) 
	+ \frac1{\sqrt N}\sum_{j=1}^\ell F_m(z_j),
	\]
	therefore, by \autoref{lem:Fbounds},
	\begin{equation}\label{eq:rate3}
	\frac{e^{\frac12\beta V_m(0,0)}}{Z_{\beta,N}}
	\Big|\mathbb{E}\Bigl[\Es^{N-k}\int_{\T^{2k}}f(\hat x)(\delta(\hat x)\,d\hat x-1)\Bigr]\Big|
	\lesssim \frac1{\sqrt N}(\log m)^{1/2}.
	\end{equation}
	It remains to consider only the term,
	\[
	\frac{1}{Z_{\beta,N}}\Biggl(e^{\frac12\beta V_m(0,0)}
	(\mathbb{E}[\Es^{N-k}]-Z_{\beta,N}\Biggr)
	\int_{\T^{2k}}f(\hat x)\,d\hat x
	+ \int_{\T^{2k}}f(\hat x)\,d\hat x
	\]
	and we wish to estimate the contribution to the rate of convergence
	of the term in brackets in the formula above.
	Applying the same estimates of above to $f\equiv1$, we
	see that the term in brackets is, up to error terms of the same order
	of those in \eqref{eq:rate1} and \eqref{eq:rate2}, controlled by
	\begin{equation}\label{eq:rate4}
	e^{\frac12\beta V_m(0,0)}(\mathbb{E}[\Es^{N-k}]-\mathbb{E}[\Es^N])
	\lesssim\frac1N\log m
	\end{equation}
	The last inequality follows from \autoref{lem:Fbounds}.
	We finally choose $m=N^a$. With $a=1+\tfrac{r}4$,
	\eqref{eq:rate1} is controlled by $N^{-1/2}(\log N)^{3/2}$,
	as well as \eqref{eq:rate3} and \eqref{eq:rate4}.
	Likewise for \eqref{eq:rate2} if we choose the
	integer $n>1+\tfrac\beta{2\pi}a$.
\end{proof}

\bibliographystyle{plain}

\begin{thebibliography}{10}
	
	\bibitem{BePiPu87}
	G.~Benfatto, P.~Picco, and M.~Pulvirenti.
	\newblock On the invariant measures for the two-dimensional {E}uler flow.
	\newblock {\em J. Statist. Phys.}, 46(3-4):729--742, 1987.
	
	\bibitem{Br78}
	David~C. Brydges.
	\newblock A rigorous approach to {D}ebye screening in dilute classical
	{C}oulomb systems.
	\newblock {\em Comm. Math. Phys.}, 58(3):313--350, 1978.
	
	\bibitem{BrFe80}
	David~C. Brydges and Paul Federbush.
	\newblock Debye screening.
	\newblock {\em Comm. Math. Phys.}, 73(3):197--246, 1980.
	
	\bibitem{clmp92}
	E.~Caglioti, P.-L. Lions, C.~Marchioro, and M.~Pulvirenti.
	\newblock A special class of stationary flows for two-dimensional {E}uler
	equations: a statistical mechanics description.
	\newblock {\em Comm. Math. Phys.}, 143(3):501--525, 1992.
	
	\bibitem{clmp95}
	E.~Caglioti, P.-L. Lions, C.~Marchioro, and M.~Pulvirenti.
	\newblock A special class of stationary flows for two-dimensional {E}uler
	equations: a statistical mechanics description. {II}.
	\newblock {\em Comm. Math. Phys.}, 174(2):229--260, 1995.
	
	\bibitem{dpz}
	Giuseppe Da~Prato and Jerzy Zabczyk.
	\newblock {\em Stochastic equations in infinite dimensions}, volume 152 of {\em
		Encyclopedia of Mathematics and its Applications}.
	\newblock Cambridge University Press, Cambridge, second edition, 2014.
	
	\bibitem{deutschlavaud}
	C.~Deutsch and M.~Lavaud.
	\newblock Equilibrium properties of a two-dimensional coulomb gas.
	\newblock {\em Phys. Rev. A}, 9:2598--2616, Jun 1974.
	
	\bibitem{DuPu82}
	D.~D\"{u}rr and M.~Pulvirenti.
	\newblock On the vortex flow in bounded domains.
	\newblock {\em Comm. Math. Phys.}, 85(2):265--273, 1982.
	
	\bibitem{Fr76}
	J\"{u}rg Fr\"{o}hlich.
	\newblock Classical and quantum statistical mechanics in one and two
	dimensions: two-component {Y}ukawa- and {C}oulomb systems.
	\newblock {\em Comm. Math. Phys.}, 47(3):233--268, 1976.
	
	\bibitem{Gr19}
	Francesco {Grotto}.
	\newblock {Essential Self-Adjointness of Liouville Operator for 2D Euler Point
		Vortices}.
	\newblock {\em arXiv e-prints}, page arXiv:1910.13134, Oct 2019.
	
	\bibitem{GrRo19}
	Francesco {Grotto} and Marco {Romito}.
	\newblock {A Central Limit Theorem for Gibbsian Invariant Measures of 2D Euler
		Equation}.
	\newblock {\em arXiv e-prints}, page arXiv:1904.01871, Apr 2019.
	
	\bibitem{gunsonpanta}
	J.~Gunson and L.~S. Panta.
	\newblock Two-dimensional neutral {C}oulomb gas.
	\newblock {\em Comm. Math. Phys.}, 52(3):295--304, 1977.
	
	\bibitem{Ke83}
	Tom Kennedy.
	\newblock Debye-{H}\"{u}ckel theory for charge symmetric {C}oulomb systems.
	\newblock {\em Comm. Math. Phys.}, 92(2):269--294, 1983.
	
	\bibitem{Ke84}
	Tom Kennedy.
	\newblock Mean field theory for {C}oulomb systems.
	\newblock {\em J. Statist. Phys.}, 37(5-6):529--559, 1984.
	
	\bibitem{lionsbook}
	Pierre-Louis Lions.
	\newblock {\em On {E}uler equations and statistical physics}.
	\newblock Cattedra Galileiana. [Galileo Chair]. Scuola Normale Superiore,
	Classe di Scienze, Pisa, 1998.
	
	\bibitem{marchioropulvirenti}
	Carlo Marchioro and Mario Pulvirenti.
	\newblock {\em Mathematical theory of incompressible nonviscous fluids},
	volume~96 of {\em Applied Mathematical Sciences}.
	\newblock Springer-Verlag, New York, 1994.
	
	\bibitem{onsager}
	L.~Onsager.
	\newblock Statistical hydrodynamics.
	\newblock {\em Nuovo Cimento (9)}, 6(Supplemento, 2 (Convegno Internazionale di
	Meccanica Statistica)):279--287, 1949.
	
	\bibitem{Sa78}
	Stu Samuel.
	\newblock Grand partition function in field theory with applications to
	sine-gordon field theory.
	\newblock {\em Phys. Rev. D}, 18:1916--1932, Sep 1978.
	
	\bibitem{Sc96}
	Steven Schochet.
	\newblock The point-vortex method for periodic weak solutions of the 2-{D}
	{E}uler equations.
	\newblock {\em Comm. Pure Appl. Math.}, 49(9):911--965, 1996.
	
\end{thebibliography}

\end{document}